\documentclass[letterpaper,12pt]{article}
\usepackage{format}
\usepackage{amssymb}
\usepackage{amsfonts}
\usepackage{enumerate}
\usepackage{multirow}
\usepackage{algorithm}
\usepackage{algorithmic}
\usepackage{cancel}

\newcommand{\mech} {{\tt MATRIX}\/}
\newcommand{\squishlisttwo}{
\begin{list}{$\blacktriangleright$}
{ \setlength{\itemsep}{0.5pt}
\setlength{\parsep}{0pt}
\setlength{\topsep}{0pt}
\setlength{\partopsep}{0.5pt}
\setlength{\leftmargin}{1em}
\setlength{\labelwidth}{1em}
\setlength{\labelsep}{0.5em} } }

\newcommand{\squishend}{
\end{list} }

\begin{document}

\title{\textbf{Dynamic Mechanism Design with Interdependent Valuations}\thanks{A preliminary version of this work was presented in the conference on Uncertainty in Artificial Intelligence, 2011.}
}

\author[1]{Swaprava Nath}
\author[2]{Onno Zoeter}
\author[3]{Y. Narahari}
\author[2]{Christopher R. Dance}

\affil[1]{\small Indian Statistical Institute, New Delhi}
\affil[2]{\small Xerox Research Centre Europe, Meylan, France}
\affil[3]{\small Indian Institute of Science, Bangalore}

\date{June 2, 2015}
\maketitle

\begin{abstract}
We consider an infinite horizon dynamic mechanism design problem with interdependent valuations. In this setting the type of each agent is assumed to be evolving according to a first order Markov process and is independent of the types of other agents. However, the valuation of an agent can depend on the types of other agents, which makes the problem fall into an interdependent valuation setting. Designing truthful mechanisms in this setting is non-trivial in view of an impossibility result which says that for interdependent valuations, any efficient and ex-post incentive compatible mechanism must be a constant mechanism, even in a static setting. \cite{mezzetti2004mechanism} circumvents this problem by splitting the decisions of allocation and payment into two stages. However, Mezzetti's result is limited to a static setting and moreover in the second stage of that mechanism, agents are weakly indifferent about reporting their valuations truthfully. This paper provides a first attempt at designing a dynamic mechanism which is efficient, {\em strict} ex-post incentive compatible and ex-post individually rational in a setting with interdependent values and Markovian type evolution.
\end{abstract}


\section{Introduction}
\label{sec_intro}


Organizations often face the problem of executing a task for which they do not have enough resources or expertise. It may also be difficult, both logistically and economically, to acquire those resources. For example, in the area of healthcare, it has been observed that there are very few occupational health professionals and doctors and nurses in all specialities at the hospitals in the UK~\citep{nicholson04occupational-health}. With the advances in computing and communication technologies, a natural solution to this problem is to outsource the tasks to experts outside the organization. Hiring experts beyond an organization was already in practice. However, with the advent of the Internet, this practice has extended even beyond the international boundaries, e.g., some U.S.\ hospitals are outsourcing the tasks of reading and analyzing scan reports to companies in Bangalore, India~\citep{AP04hospital-outsourcing}. \cite{gupta2008outsourcing} give a detailed description of how the healthcare 
industry uses the outsourcing tool.

The organizations where the tasks are outsourced (let us call them vendors) have quite varied efficiency levels. For tasks like healthcare, it is extremely important to hire the right set of experts. If the efficiency levels of the vendors and the difficulties of the medical tasks are observable by a central management (controller), and if the efficiency levels vary over time according to a Markov process, the problem of selecting the right set of experts reduces to a Markov Decision Problem (MDP), which has been well studied in the literature~\citep{bertsekas95DP-opt-control, puterman05MDP}. Let us call the efficiency levels and task difficulties together as \emph{types} of the tasks and resources.

However, the types are usually observed privately by the vendors and hospitals (agents), who are rational and intelligent. The efficiencies of the vendors are private information of the vendors (depending on what sort of doctors they hire, or machines they use), and they might misreport this information in order to win the contract and to increase their net returns. At the same time the difficulty of the medical task is private to the hospital, and is unknown to the experts. A strategic hospital, therefore, can misreport the task difficulty to the hired experts as well. Hence, the asymmetry of information at different agents' end transforms the problem from a \emph{completely} or \emph{partially} observable MDP into a \emph{dynamic game} among the agents. 

Motivated by examples of this kind, in this paper, we analyze them using a formal mechanism design framework. We consider only cases where the solution of the problem involves monetary compensation in quasi-linear form. The reporting strategy of the agents and the decision problem of the controller is dynamic since we assume that the types of the tasks and resources are varying with time. In addition, the above problem has two characteristics, namely, \emph{interdependent values}: in a selected team of agents, the valuation of an agent depends not only on her own skills but also on the skills of other selected agents, and \emph{exchange economy}: a trade environment where both buyers (task owners) and sellers (resources) are present. In this paper, the theme of modeling and analysis would be centered around the settings of task outsourcing to strategic experts. We aim to have a socially efficient mechanism, and at the same time, that would demand truthfulness and voluntary participation of the agents.

\subsection{Prior work}
\label{sec:literature}

The above properties have been investigated separately in literature on \emph{dynamic mechanism design}. \cite{bergemann-valimaki10dynamic-pivot} have proposed an efficient mechanism called the \emph{dynamic pivot mechanism}, which is a generalization of the Vickrey-Clarke-Groves (VCG) mechanism~\citep{vickrey1961counterspeculation, clarke71public-goods, groves73VCG} in a dynamic setting, and serves to be truthful and efficient. \cite{athey-segal07efficient-DMD} consider a similar setting with an aim to find an efficient mechanism that is budget balanced. \cite{cavallo-etal06opt-coordinated-planning} develop a mechanism similar to the dynamic pivot 
mechanism in a setting with agents whose type evolution follows a Markov process. In a later work, \cite{cavallo-etal09dynamic-pop-types} consider \emph{periodically inaccessible} agents and dynamic private information jointly. Even though these mechanisms work for an exchange economy, they have the underlying assumption of \emph{private values}, i.e., the reward experienced by an agent is a function of the allocation and her own private types. \cite{mezzetti2004mechanism, mezzetti07MD-dependent-values}, on the other hand, explored the other facet, namely, {\em interdependent values}, but in a static setting, and proposed a truthful mechanism. The mechanism proposed in these two papers use a two-stage mechanism, since it is impossible to design a single-stage mechanism satisfying both truthfulness and efficiency even for a static setting~\citep{jehiel-moldovanu01efficient-dependent-valuation}. However, the mechanism provides a weak truthfulness guarantee in the second stage of the game. A similar result in the setting of interdependent valuations with static types by \cite{Nath2013} ensures that the truthfulness guarantee is strict. However, since both \cite{Nath2013} and \cite{mezzetti2004mechanism} consider mechanisms that use two stages of information realization - in the first stage the types are realized and the allocation is decided, and in the second stage the valuations are realized by the agents and payments are decided - both of them require attention on how the information is revealed to the agents. In this paper, we follow an approach similar to \cite{Nath2013} that guarantees strict truthfulness. However, the equilibrium concept used here is ex-post Nash because we assume agents play in an incomplete information setting, and contrast this with the mechanism of \cite{mezzetti2004mechanism}. We also discuss how a complete information setting along with the equilibrium concept of subgame perfection plays an important role in these results. We explain this point in detail while presenting the main result of the paper.

\subsection{Contributions}
\label{sec:contribution}

In this paper, we propose a dynamic mechanism named {\bf M}DP-based {\bf A}llocation and {\bf TR}ansfer in {\bf I}nterdependent-valued e{\bf X}change economies (abbreviated \mech), which is designed to address the class of \emph{interdependent values}. It extends the results of \cite{mezzetti2004mechanism} to a dynamic setting, and with a certain allocation and valuation structure, serves as an efficient, truthful mechanism where agents receive non-negative payoffs by participating in it. The key feature that distinguishes our model and results from that of the existing dynamic mechanism literature is that we address the interdependent values and dynamically varying types (in an exchange economy) jointly and provide a {\em strict} ex-post incentive compatible mechanism.
%
%
In Table~\ref{table:contribution}, we have summarized the different paradigms of the mechanism design problem, and their corresponding solutions in the literature.
\begin{table}[h]
\centering
 \begin{tabular}{|l|p{0.3\columnwidth}|p{0.4\columnwidth}|} \hline
Valuations & STATIC & DYNAMIC \\ \hline
Independent & \textbf{VCG Mechanism} & \textbf{Dynamic Pivot Mechanism} \\ 
& \citep{vickrey1961counterspeculation, clarke71public-goods, groves73VCG} & \citep{bergemann-valimaki10dynamic-pivot, cavallo-etal06opt-coordinated-planning} \\ \hline
Interdependent & \textbf{Generalized VCG} & {\textcolor{blue}{Mechanism \mech}} \\
& \citep{mezzetti2004mechanism} & (this paper) \\ \hline
 \end{tabular}
\caption{The different paradigms of mechanism design problems with their solutions.}
\label{table:contribution}
\end{table}

Our main contributions in this paper can be summarized as follows.

\begin{itemize}
   \item We propose a dynamic mechanism \mech, that is {\em efficient}, {\em truthful} (Theorem~\ref{thm:epic}) and {\em voluntary participatory} (Theorem~\ref{thm:epir}) for the agents in an \emph{interdependent-valued exchange economy}.
  \squishlisttwo
   \item This extends the classic mechanism proposed by \citet{mezzetti2004mechanism} to a dynamic setting.
   \item It solves the issue of weak indifference by the agents in the second stage of the classic mechanism.
  \squishend
  However, we will see that Theorem~\ref{thm:epic} is true with a restricted domain of {\em subset allocation} and {\em peer-influenced valuations}. These two properties were not needed to achieve a similar claim in the {\em static} setting \citep{Nath2013}. We do not know if these are the minimal requirements for efficiency and truthfulness, but it is important to note that these properties in the {\em dynamic} setting do not immediately follow from its {\em static} counterpart.
  \item We discuss why the dynamic pivot mechanism~\citep{bergemann-valimaki10dynamic-pivot} does not satisfy all the properties that \mech\ satisfies (Section~\ref{sec:why-not-DPM}).
  \item We discuss that these results can be extended to a more general setting in Section~\ref{sec:general}.

\end{itemize}

We also discuss that \mech\ comes at a computational cost which is the same as that of its independent value counterpart (Section~\ref{sec:complexity}).

The rest of the paper is organized as follows. We introduce the formal model in Section~\ref{model-back}, and present the main results in Section~\ref{sec:main-results}. In Section~\ref{sec:general}, we discuss about a generalization of the main results. We conclude the paper in Section~\ref{concl} with some potential future works.

\section{Background and Model}
\label{model-back}

Let the set of agents be given by $N =
\{1,\ldots,n\}$, who interact with each other for a countably infinite time horizon indexed by time steps $t=0,1,2,\ldots$. The time-dependent type of each agent is denoted by $\theta_{i,t} \in
\Theta_i$ for $i \in N$. We will use the shorthands $\theta_t \equiv
(\theta_{1,t},\ldots,\theta_{n,t}) \equiv (\theta_{i,t}, \theta_{-i,t})$, where $\theta_{-i,t}$ denotes the type vector of all agents excluding agent $i$. We will refer to $\theta_t$ as the type profile at time $t$, $\theta_t \in \Theta \equiv \times_{i \in N} \Theta_i$.

The allocation set is denoted by $A$. In each round $t$, the mechanism designer chooses an allocation $a_t$ from this set and decides a payment $p_{i,t}$ to agent $i$. The allocation leads to a valuation to agent $i$, $v_i : A \times \Theta \to \mathbb{R}$. This is in contrast to the classical independent valuations (also called \emph{private values}) case where valuations are assumed to depend only on $i$'s own type; $v_i : A \times \Theta_i \rightarrow \mathbb{R}$. However, we assume for all $i$, $|v_i(a,\theta)| \leq M < \infty$, for some $M \in \mathbb{R}$ and for all $a$ and $\theta$.


  
\paragraph{\sc Stationary Markov Type Transitions, SMTT} 

The combined type $\theta_t$ follows a first order Markov process
which is governed by the transition probability function $F(\theta_{t+1}|a_t, \theta_t)$, which is independent across agents, where $a_t$ is the allocation at period $t$.

\begin{definition}[Stationary Markov Type Transitions, SMTT] \label{SMTT}
 We call the type  transitions to follow stationary Markov type transitions if the joint distribution $F$ of the types of the agents $\theta_t \equiv (\theta_{1,t}, \cdots, \theta_{n,t})$, and the marginals $F_i$'s exhibit the following for all $t$.
 \begin{equation}
  \label{eq:MTT}
   \begin{aligned}
    F(\theta_{t+1}|a_t, \theta_t, \theta_{t-1}, \cdots, \theta_0) &= F(\theta_{t+1}|a_t, \theta_t), \mbox{ and } \\
    F(\theta_{t+1}|a_t, \theta_t) &= \prod_{i \in N} F_i(\theta_{i, t+1}|a_t, \theta_{i,t}).
   \end{aligned}
 \end{equation}
\end{definition}

We will assume the types to follow SMTT throughout this paper.

For an easier exposition of the more general properties that lead to the same conclusions as in this paper, we will restrict our attention to a restricted space of allocations and valuations. In Section~\ref{sec:general}, we comment on the generalization of our results by introducing certain assumptions that subsume the following two assumptions on the allocation and valuations.

\paragraph{\sc Subset Allocation, SA} 

Let us motivate this restriction with the medical task assignment example given in the previous section. The organizations outsource tasks to experts for a payment, where the expert may have different and often time-varying capabilities of executing the task. The task owners come with a specific task difficulty (type of the task owner), which is usually privately known to them, while the workers' capabilities (types of the workers) are their private information. A central planner's job in this setting is to efficiently assign the tasks to a group of workers. Clearly, in this setting, the set of possible allocations is the set of the subsets of agents, i.e., $A = 2^N$. Note that, for a finite set of players, the allocation set is always finite. So, we can formally define this setting as follows.

\begin{definition}[Subset Allocation, SA] \label{SA}
 When the set of allocations is the set of all subsets of the agent set, i.e., $A = 2^N$, we call the domain a {\em subset allocation} domain. Similarly, $A_{-i} = 2^{N \setminus \{i\}}$ denotes the set of allocations excluding agent $i$.
\end{definition}


\paragraph{\sc Peer Influenced Valuations, PIV} 

Even though the valuation of agent $i$ is affected by not only her private type but also by the types of others, it is often the case that the valuation is affected by the types (e.g. the efficiencies of the workers in a joint project) of only the {\em selected} agents. The valuation therefore is a function of the types of the allocated agents and not the whole type vector. We also assume that the value of a non-selected agent is zero. The set of valuations satisfying the above two conditions is called the set of {\em peer influenced valuations} (PIV).

\begin{definition}[Peer Influenced Valuations, PIV] \label{PIV}
 This is a special set of interdependent valuations in the SA domain, where the valuation of agent $i$ is a function of the types of other selected agents, given by,
 \begin{align}
 v_i(a, \theta) &= \left \{ 
		  \begin{array}{ll}
		   v_i(a, \theta_a) & \mbox{ if } i \in a \\
		   0 & \mbox{ otherwise,}
		  \end{array}
		  \right.
\end{align}
where $\theta_a \in \times_{i \in a} \Theta_i$, for an allocation $a \in A = 2^N$.
\end{definition}

The properties SA and PIV together allow for a well-defined counterfactual social welfare in a world where a particular agent does not exist. See also Equation~(\ref{eq:social-welfare-except-i-PIV}).

\paragraph{\sc Efficient Allocation, EFF} 

The mechanism designer aims to maximize the 
sum of the valuations of task owners and workers, summed over an infinite horizon,
geometrically discounted with factor $\delta \in (0,1)$. The discount factor accounts for the fact that a future payoff is less valued by an agent than a current stage payoff. We assume $\delta$ to be common knowledge.
If the designer would have perfect information about the $\theta_t$'s, his
objective would be to find a {\em policy} $\pi_t$, which is a sequence of allocation functions from time $t$, that yields the following for all $t$ and for all type profiles $\theta_t$,
 \begin{equation}
  \label{eq:eff-allocation-old}
  \pi_t \in \argmax_{\gamma} \ \mathbb{E}_{\gamma, \theta_t} \left[
\sum_{s=t}^{\infty} \delta^{s-t} \sum_{i \in N} v_i(a_s(\theta_{s}), \theta_{s})\right],
 \end{equation}
where $\gamma = (a_t(\cdot), a_{t+1}(\cdot), \ldots)$ is any arbitrary sequence of allocation functions.
Here we use $\mathbb{E}_{\gamma, \theta_t} [ \cdot ] = \mathbb{E} [\ \cdot \ | \theta_t; \gamma]$ for brevity of notation. We point to the fact that the allocation policy $\gamma$ is not a random variable in this expectation computation. The policy is a functional that specifies what action to take in each time instant for a given type profile. Different policies will lead to different sequences of allocation functions over the infinite horizon, and the efficient allocation is the one that maximizes the expected discounted sum of the valuations of all the agents.

In general, the allocation policy $\pi_t$ depends on the time instant $t$. However, for the special kind of stochastic behavior of the type vectors, namely SMTT, and due to the infinite horizon discounted utility, this policy becomes stationary, i.e., independent of $t$. We will denote such a stationary policy by $\pi = (a(\cdot), a(\cdot), \ldots)$.
Thus, the efficient allocation under SMTT reduces to solving for the optimal action in the following stationary Markov Decision Problem (MDP).

\begin{eqnarray} 
W(\theta_t) &=& \max_{\pi} \ \mathbb{E}_{\pi, \theta_t} \left[ \sum_{s = t}^{\infty}
\delta^{s-t} \sum_{j \in N} v_j(a(\theta_{s}), \theta_{s})\right] \nonumber \\ 
&=& \max_{a \in A} \ \mathbb{E}_{a, \theta_t} \left[ \sum_{j \in N} v_j(a,
\theta_{t}) + \delta \mathbb{E}_{\theta_{t+1} | a, \theta_t} W(\theta_{t+1})
\right]. \label{eq:social-welfare}
\end{eqnarray}

Here, with a slight abuse of notation, we have used $a$ to denote the actual action taken in $t$ rather than the allocation function.
The second equality comes from a standard recursive argument for stationary infinite horizon MDPs. We refer an interested reader to standard text \citep[e.g.]{puterman05MDP} for this reduction and the general properties of MDPs. We have used the following shorthand, $\mathbb{E}_{\theta_{t+1} | a, \theta_t} [ \cdot ] = \sum_{\theta_{t+1}} p(\theta_{t+1} | \theta_t; a_t)[ \cdot ]$.
We will refer to $W$ as the {\em social welfare}. The efficient allocation under SMTT is defined as follows.

\begin{definition}[Efficient Allocation, EFF] \label{EFF}
 An allocation policy $a(\cdot)$ is \textbf{efficient} under SMTT if for all type profiles $\theta_t$,
 \begin{equation}
  \label{eq:eff-allocation}
a(\theta_t) \in \argmax_{a \in A} \ \mathbb{E}_{a, \theta_t} \left [ \sum_{j \in N} v_j(a,
\theta_t) + \delta \mathbb{E}_{\theta_{t+1} | a, \theta_t} W(\theta_{t+1}),
\right ].
 \end{equation}
\end{definition}


\paragraph{\sc Challenges in mechanism design with interdependent valuations}
    
The value interdependency among the agents poses a challenge for designing mechanisms. Even in a static setting, if the allocation and payment are decided simultaneously under the interdependent valuation setting, efficiency and Bayesian incentive compatibility (and therefore ex-post incentive compatibility) cannot be satisfied together~\citep{jehiel-moldovanu01efficient-dependent-valuation}. In a later paper, \citet{jehiel2006limits} show that the only deterministic social choice functions that are ex-post implementable in generic mechanism design frameworks with multi-dimensional signals, interdependent valuations, and transferable utilities are constant functions. In view of these impossibility results, we are compelled to split the decisions of allocation and payment in two separate stages. We would mimic the two-stage mechanism of \cite{mezzetti2004mechanism} for each time instant of the dynamic setting (see Figure~\ref{fig:example-2}). We consider a direct revelation mechanism. In the first stage of this two-stage mechanism, the agents observe their individual types $\theta_{i,t} \in \Theta_i, \ i \in N$. The strategies available to the agents are to report any type $\hat{\theta}_{i,t} \in \Theta_i$. The designer decides the allocation $a(\hat{\theta}_t)$ depending on the reported types $\hat{\theta}_t$ in first stage. The reported types of the agents are {\em not} revealed publicly in the first stage. This assumption plays a crucial role in the concept of {\em incentive compatibility} we use in this paper. We discuss this after the definition of incentive compatibility briefly and in detail in the next section. After the allocation, the agents observe their valuations $v_i(a(\hat{\theta}_t), \theta_t)$'s, and report $\hat{v}_{i,t}$'s to the designer. The payment decision is made after this second stage of reporting. Our definition of incentive compatibility is accordingly modified for a two stage mechanism.

\begin{figure}[t!]
 \centering
 \includegraphics[width=\columnwidth]{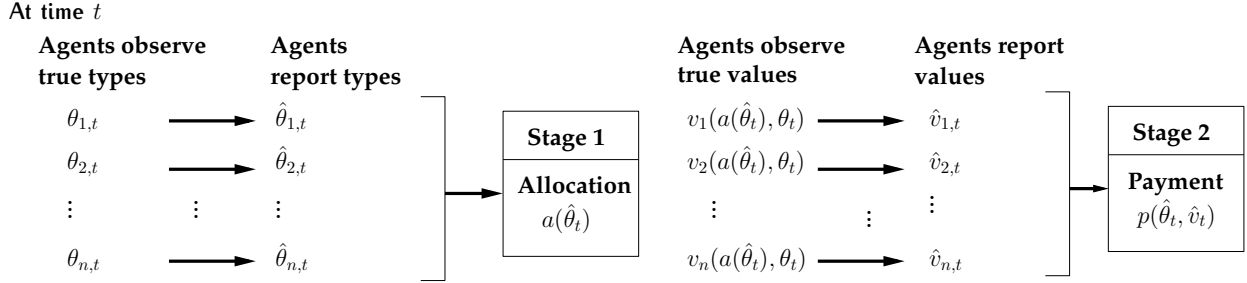}
 \caption{Graphical illustration of a candidate dynamic mechanism in an interdependent value setting.}
 \label{fig:example-2}
\end{figure}

Due to SMTT and the infinite horizon of the MDP, we will focus only on stationary mechanisms, that give a stationary allocation and payment to the agents in each round of the dynamic game.
Let us denote a typical two-stage dynamic mechanism by $M = \langle a, p \rangle$. The function $a : \Theta \to A$ yields an allocation for a reported type profile $\hat{\theta}_t$ in round $t$. 
Depending on the reported types in the first stage, the mechanism designer decides the allocation $a(\hat{\theta}_t)$, due to which agent $i$ experiences a valuation of $v_i(a(\hat{\theta}_t), \theta_t)$ in round $t$. Let us suppose that in the second stage, the reported value vector is given by $\hat{v}_t$. The payment function $p$ is a vector where $p_i(\hat{\theta}_t, \hat{v}_t)$ is the payment received by agent $i$ at instant $t$. 
Combining the value and payment in each round we can write the expected discounted utility of agent $i$ in the \emph{quasi-linear} setting, denoted by $u^M_i(\hat{\theta}_t, \hat{v}_t | \theta_t)$, when the true type vector is $\theta_t$ and the reported type and value vectors are $\hat{\theta}_t$ and $\hat{v}_t$ respectively. This utility has two parts: (a) the current round utility, and (b) expectation over the future round utilities. The expectation over the future rounds is taken on the true types. Thus the effect of manipulation is limited only to the current round in this utility expression. This is enough to consider due to the {\em single deviation principle} of \citet{blackwell1965discounted}.

\begin{align}
  \lefteqn{u^M_i(\hat{\theta}_t, \hat{v}_t | \theta_t)} \nonumber \\
  &= \underbrace{v_i(a(\hat{\theta}_t), \theta_{t}) +
p_i(\hat{\theta}_t, \hat{v}_{t})}_{\textsf{current round utility}} + \underbrace{\mathbb{E}_{\pi, \theta_t}
\left[\sum_{s= t+1}^{\infty} \delta^{s-t} (v_i(a(\theta_s),\theta_{s}) + p_i(\theta_s,v_s)) \right]}_{\textsf{expected discounted future utility}} \label{eq:discounted-utility}
\end{align}
Here $\pi$ denotes the stationary policy of actions, $(a(\cdot),a(\cdot), \ldots)$. For the SMTT, the type evolution is dependent on only the current type profile and action. To avoid confusion, we will use $\pi$, $a(\hat{\theta}_t)$, or $a(\theta_s), s \geq t+1$, according to the context.

Equipped with this notation, we can now define {\em incentive compatibility}.

\begin{definition}[w.p. EPIC] \label{EPIC}
 A mechanism $M = \langle a, p \rangle$ is {\em within period Ex-post Incentive Compatible (w.p. EPIC)} if for all agents $i \in N$, for all possible true types $\theta_t$, for all reported types $\hat{\theta}_{i,t}$, for all reported values $\hat{v}_{i,t}$, and for all $t$,
 \begin{align*}
  \label{eq:epic}
  \lefteqn{u^M_i(\theta_t, (v_i(a(\theta_t), \theta_t), v_{-i}(a(\theta_t), \theta_t) ) | \theta_t)} \qquad \\
     \qquad &\geq u^M_i((\hat{\theta}_{i,t}, \theta_{-i,t}), (\hat{v}_{i,t}, v_{-i}(a(\hat{\theta}_{i,t}, \theta_{-i,t}), \theta_t)) | \theta_t)
 \end{align*}
\end{definition}

That is, reporting the types and valuations in the two stages truthfully is an ex-post Nash equilibrium. We use `ex-post' to denote that the agent chooses her action after observing her own type and valuation, and not the types of others, since that is not revealed to her according to the mechanism considered here.~\footnote{Some readers may interpret the term `ex-post' differently, since the term is conventionally used in the context of single stage mechanisms, i.e., where the decisions of allocation and transfer are decided simultaneously (see, e.g., \cite{jehiel2006limits}) and it denotes that truthful reporting is optimal for every realization of the other agents' types even if the agent knew the other agents' types. In the context of two-stage mechanisms that we consider here, we feel that this equilibrium of full observability can be better called as `subgame perfect' equilibrium. This is the equilibrium concept used in the static two stage mechanism by \citet{mezzetti2004mechanism}, and we discuss in detail the difference of the two equilibria concepts in the next section.} 
The reported valuation $\hat{v}_{i,t}$ is therefore a function of the types $\theta_{i,t}$ and $\hat{\theta}_{i,t}$ and not of either $\theta_{-i,t}$ and $\hat{\theta}_{-i,t}$, according to the assumption above. An interesting question would be: what happens when the agents' type reports in the first stage are made public. The agents' valuation reports in the second stage can then depend on the type reports in the first stage. The appropriate equilibrium concept in that setting is the {\em subgame perfect equilibrium}. We present a detailed discussion on the implications of revealing the type reports in the first stage after presenting the proposed mechanism in the next section.

In this context, {\em individual rationality} is defined as follows.

\begin{definition}[w.p. EPIR] \label{EPIR}
 A mechanism $M = \langle a, p \rangle$ is {\em within period Ex-post Individually Rational (w.p. EPIR)} if for all agents $i \in N$, for all possible true types $\theta_{t}$ and for all $t$,
\[u^M_i(\theta_t, (v_i(a(\theta_t), \theta_t), v_{-i}(a(\theta_t), \theta_t)) | \theta_t) \geq 0.\]
\end{definition}
That is, reporting the types and valuations in the two stages truthfully yields non-negative expected utility.

\section{The \mech\ Mechanism under SA and PIV}
\label{sec:main-results}

In the interdependent valuation setting, our goal is to design a mechanism which is efficient (Def.~\ref{EFF}), w.p. EPIC (Def.~\ref{EPIC}), and w.p. EPIR (Def.~\ref{EPIR}). This is non-trivial because to achieve efficient allocation in a dynamic setting, one needs to consider the expected future evolution of the types of the agents, which would reflect in the allocation and payment decisions, and for this reason a {\em fixed payment} mechanism or a {\em repeated VCG} fails to satisfy efficiency (Def.~\ref{EFF}). The value interdependency among the agents plays a crucial role here. Even in a static  interdependent value setting, if the allocation and payment are decided simultaneously, one cannot guarantee efficiency and incentive compatibility together \citep{jehiel-moldovanu01efficient-dependent-valuation}. One way out is to split the decision of allocation and payment in two stages~\citep{mezzetti2004mechanism}. 

Following this observation, we propose {\bf M}DP-based {\bf A}llocation and {\bf TR}ansfer in {\bf I}nterdependent-valued e{\bf X}change economies (\mech), which we prove to satisfy EFF, w.p. EPIC and w.p. EPIR under the restricted setting of SA and PIV.


%
%
%


Given the dynamics of the game, illustrated in Figure~\ref{fig:example-2}, the agents report their types in the first stage, and then the allocation is decided. In the second stage, they report their experienced values and the payment is decided. The task of the mechanism designer, therefore, is to design the allocation and payment rules $\langle a, p \rangle$ in each time instant. 



In the context of SA and PIV, the social welfare given by Eq.~(\ref{eq:social-welfare}) is modified as follows.

\begin{eqnarray} 
W(\theta_t)
&=& \max_{\pi} \ \mathbb{E}_{\pi, \theta_t} \left[ \sum_{s = t}^{\infty}
\delta^{s-t} \sum_{j \in N} v_j(a(\theta_s), \theta_{a(\theta_s)})\right] \nonumber \\ 
&=& \max_{a \in A} \ \mathbb{E}_{a, \theta_t} \left[ \sum_{j \in N} v_j(a,
\theta_{a}) + \delta \mathbb{E}_{\theta_{t+1} | a, \theta_t} W(\theta_{t+1})
\right]. \label{eq:social-welfare-PIV}
\end{eqnarray}

We also define the maximum social welfare excluding agent $i$ to be $W_{-i}(\theta_{-i,t})$, which is the same as Eq.~(\ref{eq:social-welfare}) except now the sum of the valuations and the allocations are over all agents $j \neq i$. We also use the set of allocations excluding $i$ to be $A_{-i}$ as defined by SA,

\begin{eqnarray} 
\lefteqn{W_{-i}(\theta_{-i,t})} \nonumber \\ 
&=& \max_{a_{-i} \in A_{-i}} \ \mathbb{E}_{a_{-i}, \theta_t} \left[ \sum_{j \in N \setminus \{i\}} v_j(a_{-i},
\theta_{a_{-i}}) + \delta \mathbb{E}_{\theta_{t+1} | a_{-i}, \theta_t} W_{-i}(\theta_{-i,t+1})
\right]. \label{eq:social-welfare-except-i-PIV}
\end{eqnarray}
Note that, SA and PIV are crucial for defining this quantity.
Also, when $i$ is absent, the following two notations are equivalent: $\mathbb{E}_{\theta_{t+1} | a_{-i}, \theta_t} [ \cdot ] = \mathbb{E}_{\theta_{-i,t+1} | a_{-i}, \theta_{-i,t}} [ \cdot ]$, since the type of $i$ will be unchanged when she is not in the game. However, we adopt the former for consistency in notation. Using the definitions above and in the previous section, now we formally present \mech.

\begin{mechanism}[\mech]
 \label{gdpm}
Given the reported type profile $\hat{\theta}_t$ in stage 1, choose the agents
$a^*(\hat{\theta}_t)$ as follows. 

\begin{equation}
\label{allocation-MATRIX}
a^*(\hat{\theta}_t) \in \argmax_{a \in A} \mathbb{E}_{a, \hat{\theta}_t} \left [ \sum_{j \in N} v_j(a,
\hat{\theta}_{a}) + \delta \mathbb{E}_{\theta_{t+1} | a, \hat{\theta}_t} W(\theta_{t+1}),
\right ] 
\end{equation}

and transfer to agent
$i$ after agents report $\hat{v}_{t}$ in stage 2, a payment of,



\begin{align}
\label{payment-MATRIX}
 p_i^*(\hat{\theta}_t, \hat{v}_{t})
 = & \left(\sum_{j \neq i} \hat{v}_{j,t}\right) + \delta \mathbb{E}_{\theta_{t+1} | a^*(\hat{\theta}_t),\hat{\theta}_t} W_{-i}(\theta_{-i,t+1}) - W_{-i}(\hat{\theta}_{-i,t}) \nonumber \\
  & \qquad - \left ( \hat{v}_{i,t} - v_i(a^*(\hat{\theta}_t), \hat{\theta}_{a^*(\hat{\theta}_t)}) \right )^2 .
\end{align}

\end{mechanism}
\begin{wrapfigure}{r}{0.65\columnwidth}
    \begin{minipage}{0.65\columnwidth}
\begin{algorithm}[H]
\caption{\mech}
\label{matrix-alg}
\begin{algorithmic}
\FORALL{time instants $t$} 
\STATE \textbf{Stage 1:}
 \FOR{agents $i=0,1,\dots,n$}
  \STATE agent $i$ observes $\theta_{i,t}$;
  \STATE agent $i$ reports $\hat{\theta}_{i,t}$;
 \ENDFOR
\STATE compute allocation $a^*(\hat{\theta}_t)$ according to Eq.~\ref{allocation-MATRIX};
\STATE \textbf{Stage 2:}
 \FOR{agents $i=0,1,\dots,n$}
  \STATE agent $i$ observes $v_i(a^*(\hat{\theta}_t), \theta_{a^*(\hat{\theta}_t)})$;
  \STATE agent $i$ reports $\hat{v}_{i,t}$;
 \ENDFOR
\STATE compute payment to agent $i$, $p_i^*(\hat{\theta}_t, \hat{v}_{t})$, Eq.~\ref{payment-MATRIX};
\STATE types evolve $\theta_t \rightarrow \theta_{t+1}$ according to SMTT;
\ENDFOR
\end{algorithmic}
\end{algorithm}
    \end{minipage}
  \end{wrapfigure}

The last quadratic term in the above equation is agent $i$'s penalty of not being consistent with the first stage report. 
The intuition of charging a penalty is to make sure that agent $i$ be consistent with her reported type $\hat{\theta}_{i,t}$ in the first stage and her value report $\hat{v}_{i,t}$ in the second stage, given that others are reporting their types and values truthfully. We will argue that when all agents other than agent $i$ reports their types and values truthfully in the two stages of the mechanism, it is the best response for agent $i$ to do so as well. This term distinguishes our mechanism from that given by \citet{mezzetti2004mechanism}, where the agents are weakly indifferent between reporting true and false values in the second stage. We summarize the dynamics of \mech\ using an algorithmic flowchart in Algorithm~\ref{matrix-alg}. 

We have used this quadratic term for the ease of exposition. However, it is easy to show that any non-negative function $g(x,\ell)$ having the property that $g(x,\ell) = 0 \Leftrightarrow x = \ell$ would still satisfy the claims made in this paper. \cite{Nath2013} use a similar term to ensure strict truthfulness in the second stage of a two stage static mechanism with interdependent valuations. 

\paragraph{\sc \mech\ and Subgame Perfection}

Since this paper is an extension of the results of \cite{Nath2013} to a dynamic type setting, we can do similar comparisons of properties with the mechanism of \cite{mezzetti2004mechanism} (let us call this the {\em classic} mechanism).
If we consider the case where the first stage type reports are made public by the mechanism, i.e., observable by all agents, then the agents have a chance of modifying their next stage report depending on that information. The concept of truthfulness should be modified to subgame perfect equilibrium in this context, which ensures that truth-telling is an equilibrium in every subgame of the two stage game. It can be shown that an agent $i$ can misreport her type in the first stage from $\theta_i$ to $\hat{\theta}_i$ when other agents are reporting their types truthfully and in this subgame, since the reported types are public, each agent's best response would be to report valuations consistent with the first stage's {\em reported} types $v_i(a^*(\hat{\theta}_t), \hat{\theta}_{a^*(\hat{\theta}_t)})$ (and not the true valuations $v_i(a^*(\hat{\theta}_t), \theta_{a^*(\hat{\theta}_t)})$), which results in more utility to agent $i$ than reporting types truthfully in the first stage (see \cite{Nath2013}, where Example 1 illustrates this and can be modified in the dynamic setting for a similar conclusion). Hence, if the first stage reports are made public, \mech\ does {\em not} ensure truthfulness in a subgame perfect equilibrium.
The {\em classic} mechanism, on the other hand, continues to satisfy truth-telling in a subgame perfect equilibrium even in this complete information scenario, and this is because the utility of the agent is unaffected by her second stage valuation reports. So, to summarize, in the incomplete information setting, \mech\ provides a strict truthfulness guarantee in the second stage and the truthfulness is in an ex-post Nash equilibrium, but in a complete information setting, it does not ensure truthfulness in a subgame perfect equilibrium, while the {\em classic} mechanism is not {\em strictly} truthful in an ex-post Nash equilibrium for an incomplete information setting, but is {\em weakly} truthful in a subgame perfect equilibrium in the complete information setting. It is important to note that even though the {\em classic} mechanism is weakly truthful in the second stage, and every agent's utility is unaffected by their valuation report, the truthfulness in the type reports in the first stage requires that the agents be truthful in the second stage. Hence, one needs to {\em assume} in the mechanism by \cite{mezzetti2004mechanism} that the agents report their valuations truthfully even when their utilities are unaffected by their reports.



%

\subsection{Efficiency and incentive compatibility}


%

The following theorem shows that \mech\ satisfies two desirable properties in the unrestricted setting.


\begin{theorem}
 \label{thm:epic}
 Under SMTT, with SA and PIV, \mech\ is EFF and w.p.\ EPIC. In addition, the second stage of \mech\ is strictly EPIC.
\end{theorem}

\mech\ is a two stage mechanism, and we need to ensure that truth-telling is a best response in both these stages.

\begin{proof}
Clearly, given true reported types, the allocation of \mech\ is efficient by Definition~\ref{EFF}. Hence, we need to show only that \mech\ is w.p.\ EPIC.

%
 
 To show that \mech\ is w.p.\ EPIC, let us assume that the true type profile at time $t$ is $\theta_t$, and all agents $j \neq i$ report their true types and values in each round $s = t, t+1, \cdots$ etc. Only agent $i$ reports $\hat{\theta}_{i,t}$ and $\hat{v}_{i,t}$ in the two stages. Therefore, $\hat{\theta}_t = (\hat{\theta}_{i,t}, \theta_{-i,t})$ and $\hat{v}_{j,t} = v_j(a^*(\hat{\theta}_t), \theta_{a^*(\hat{\theta}_t)})$, for all $j \neq i$. Using the {\em single deviation principle} \citep{blackwell1965discounted}, we conclude that it is enough to consider only a single shot deviation from the true report of the type. Hence, without loss of generality, let us assume that agent $i$ deviates only in round $t$ of this game.
 
Let us write down the discounted utility to agent $i$ at time $t$. 
\begin{eqnarray*}
 \lefteqn{u^{\text \mech}_i((\hat{\theta}_{i,t}, \theta_{-i,t}), (\hat{v}_{i,t}, v_{-i}(a^*(\hat{\theta}_{i,t}, \theta_{-i,t}), \theta_{a^*(\hat{\theta}_{i,t}, \theta_{-i,t})})) | \theta_t)} \\ 
 &=& \underbrace{v_i(a^*(\hat{\theta}_t), \theta_{a^*(\hat{\theta}_t)}) +
p_i^*(\hat{\theta}_t, \hat{v}_{t})}_{\textsf{current round utility}} \\
&& \qquad \qquad + \ \underbrace{\mathbb{E}_{\pi^*, \theta_t}
\left[\sum_{s= t+1}^{\infty} \delta^{s-t} (v_i(a^*(\theta_{s}),\theta_{a^*(\theta_{s})}) + p_i^*(\theta_s,v_s)) \right]}_{\textsf{expected discounted future utility}} 
\\
 &=& v_i(a^*(\hat{\theta}_t), \theta_{a^*(\hat{\theta}_t)}) + \sum_{j \neq i}
\hat{v}_{j,t} + \delta  \mathbb{E}_{\theta_{t+1} |
a^*(\hat{\theta}_t),\hat{\theta}_t}  W_{-i}(\theta_{-i,t+1}) - W_{-i}(\hat{\theta}_{-i,t}) \\ 
&& \qquad - \left ( \hat{v}_{i,t} - v_i(a^*(\hat{\theta}_t), \hat{\theta}_{a^*(\hat{\theta}_t)}) \right )^2 \\
&& \qquad \qquad + \ \mathbb{E}_{\pi^*, \theta_t}
\left[\sum_{s= t+1}^{\infty} \delta^{s-t} (v_i(a^*(\theta_{s}),\theta_{a^*(\theta_{s})}) + p_i^*(\theta_s,v_s)) \right] 
\end{eqnarray*}

 We use the shorthand $\pi^*$ to denote the allocation policy under \mech. This gives rise to the allocations $a(\cdot)$ in each round given the type profiles (either reported or true). The first equality is from Eq.~(\ref{eq:discounted-utility}). 
The second equality comes by substituting the expression of payment from Eq.~(\ref{payment-MATRIX}).

Now, from the previous discussion on the $\hat{v}_{j,t}$'s and $\hat{\theta}_{j,t}$'s, $j \neq i$, we get,
\begin{eqnarray}
 \lefteqn{u^{\text \mech}_i((\hat{\theta}_{i,t}, \theta_{-i,t}), (\hat{v}_{i,t}, v_{-i}(a^*(\hat{\theta}_{i,t}, \theta_{-i,t}), \theta_{a^*(\hat{\theta}_{i,t}, \theta_{-i,t})})) | \theta_t)} \nonumber \\ 
 &=& v_i(a^*(\hat{\theta}_t), \theta_{a^*(\hat{\theta}_t)}) + \sum_{j \neq i}
v_j(a^*(\hat{\theta}_t), \theta_{a^*(\hat{\theta}_t)}) + \delta  \mathbb{E}_{\theta_{t+1} | a^*(\hat{\theta}_t),\hat{\theta}_t} 
W_{-i}(\theta_{-i,t+1}) \nonumber \\ && - W_{-i}(\theta_{-i,t}) 
 - \left ( \hat{v}_{i,t} - v_i(a^*(\hat{\theta}_t), \hat{\theta}_{a^*(\hat{\theta}_t)}) \right )^2 \nonumber \\ && + \ \mathbb{E}_{\pi^*, \theta_t}
\left[\sum_{s= t+1}^{\infty} \delta^{s-t} (v_i(a^*(\theta_{s}),\theta_{a^*(\theta_{s})}) + p_i^*(\theta_s,v_s)) \right] \nonumber \\
&\leq& v_i(a^*(\hat{\theta}_t), \theta_{a^*(\hat{\theta}_t)}) + \sum_{j \neq i}
v_j(a^*(\hat{\theta}_t), \theta_{a^*(\hat{\theta}_t)}) + \delta  \mathbb{E}_{\theta_{t+1} | a^*(\hat{\theta}_t),\hat{\theta}_t} 
W_{-i}(\theta_{-i,t+1}) \nonumber \\ 
 && - W_{-i}(\theta_{-i,t}) + \ \mathbb{E}_{\pi^*, \theta_t}
\left[\sum_{s= t+1}^{\infty} \delta^{s-t} (v_i(a^*(\theta_{s}),\theta_{a^*(\theta_{s})}) + p_i^*(\theta_s,v_s)) \right] \label{eqn:exp-utility-2}
\end{eqnarray}

The equality comes because of the assumption that all agents $j \neq i$ report their types and values truthfully. The inequality is because we are ignoring a non-positive term.
Now, let us consider the last term of the above equation. 

\begin{eqnarray*}
 \lefteqn{\mathbb{E}_{\pi^*, \theta_t}
\left[\sum_{s= t+1}^{\infty} \delta^{s-t} (v_i(a^*(\theta_{s}),\theta_{a^*(\theta_{s})}) + p_i^*(\theta_s,v_s)) \right]} \\ 
  &=& \mathbb{E}_{\pi^*, \theta_t} \Big[\sum_{s= t+1}^{\infty} \delta^{s-t} \Big(v_i(a^*(\theta_{s}),\theta_{a^*(\theta_{s})}) + \sum_{j \neq i}
v_j(a^*(\theta_{s}),\theta_{a^*(\theta_{s})}) \\
  && \qquad \qquad + \ \delta  \mathbb{E}_{\theta_{s+1} | a^*(\theta_{s}),\theta_s} 
W_{-i}(\theta_{-i,s+1}) - W_{-i}(\theta_{-i,s}) \Big) \Big] \\
&=& \mathbb{E}_{\pi^*, \theta_t} \Big[\sum_{s= t+1}^{\infty} \delta^{s-t} \Big(\sum_{j \in N}
v_j(a^*(\theta_{s}),\theta_{a^*(\theta_{s})}) \\ && \qquad \qquad  + \ \delta  \mathbb{E}_{\theta_{s+1} | a^*(\theta_{s}),\theta_s} 
W_{-i}(\theta_{-i,s+1}) - W_{-i}(\theta_{-i,s}) \Big) \Big] 
\end{eqnarray*}

The first equality comes from Eq.~(\ref{payment-MATRIX}). We can now rearrange the expectation for the first term above using the Markov property of $\theta_t$ that gives, $\mathbb{E}_{\pi^*, \theta_t}[\cdot] = \mathbb{E}_{\theta_{t+1} | a^*(\hat{\theta}_t),\theta_t} [\mathbb{E}_{\pi^*, \theta_{t+1}} [\cdot]]$.
%
Therefore,

\begin{eqnarray}
\lefteqn{\mathbb{E}_{\pi^*, \theta_t}
\left[\sum_{s= t+1}^{\infty} \delta^{s-t} (v_i(a^*(\theta_{s}),\theta_{a^*(\theta_{s})}) + p_i^*(\theta_s,v_s)) \right]} \nonumber \\
&=& \mathbb{E}_{\theta_{t+1} | a^*(\hat{\theta}_t),\theta_t} \left[ \mathbb{E}_{\pi^*, \theta_{t+1}} \left(\sum_{s= t+1}^{\infty} \delta^{s-t} \sum_{j \in N} v_j(a^*(\theta_{s}),\theta_{a^*(\theta_{s})}) \right) \right] \nonumber \\ 
&& + \ \mathbb{E}_{\pi^*, \theta_t} \left[ \sum_{s= t+1}^{\infty} \delta^{s-t} \left( \delta  \mathbb{E}_{\theta_{s+1} | a^*(\theta_{s}),\theta_s} W_{-i}(\theta_{-i,s+1}) - W_{-i}(\theta_{-i,s}) \right) \right] \nonumber \\
&=& \mathbb{E}_{\theta_{t+1} | a^*(\hat{\theta}_t),\theta_t} \left( \delta W(\theta_{t+1}) \right) \nonumber \\
&& + \ \mathbb{E}_{\pi^*, \theta_t} \left[ \sum_{s= t+1}^{\infty} \delta^{s-t} \left( \delta \mathbb{E}_{\theta_{s+1} | a^*(\theta_{s}),\theta_s} W_{-i}(\theta_{-i,s+1}) - W_{-i}(\theta_{-i,s}) \right) \right] \label{eqn:exp-utility-3}
\end{eqnarray}

 The last equality comes from the definition of $W(\theta_{t+1})$. Let us now focus on the last term of the above equation.


\begin{eqnarray} \label{eqn:exp-utility-4}
 \lefteqn{\mathbb{E}_{\pi^*, \theta_t} \left[ \sum_{s= t+1}^{\infty} \delta^{s-t} \left( \delta  \mathbb{E}_{\theta_{s+1} | a^*(\theta_{s}),\theta_s} W_{-i}(\theta_{-i,s+1}) - W_{-i}(\theta_{-i,s}) \right) \right] } \nonumber \\
 &=& \cancel{\delta^2 \mathbb{E}_{\pi^*, \theta_t} W_{-i}(\theta_{-i,t+2})} - \delta \mathbb{E}_{\theta_{t+1} | a^*(\hat{\theta}_t),\theta_t} W_{-i}(\theta_{-i,t+1}) \nonumber \\
 && \qquad + \  \cancel{\delta^3 \mathbb{E}_{\pi^*, \theta_t} W_{-i}(\theta_{-i,t+3})} - \cancel{\delta^2 \mathbb{E}_{\pi^*, \theta_t} W_{-i}(\theta_{-i,t+2})} \nonumber \\
 && \qquad \qquad + \ \cdots - \cancel{\delta^3 \mathbb{E}_{\pi^*, \theta_t} W_{-i}(\theta_{-i,t+3})} \nonumber \\
 && \qquad \qquad \qquad + \ \cdots - \cdots \nonumber \\
 &=& - \delta \mathbb{E}_{\theta_{t+1} | a^*(\hat{\theta}_t),\theta_t} W_{-i}(\theta_{-i,t+1})
\end{eqnarray}

Combining Equations~\ref{eqn:exp-utility-2}, \ref{eqn:exp-utility-3}, and \ref{eqn:exp-utility-4}, we get,
\begin{eqnarray} \label{eqn:exp-utility-final}
 \lefteqn{u^{\text \mech}_i((\hat{\theta}_{i,t}, \theta_{-i,t}), (\hat{v}_{i,t}, v_{-i}(a^*(\hat{\theta}_{i,t}, \theta_{-i,t}), \theta_{a^*(\hat{\theta}_{i,t}, \theta_{-i,t})})) | \theta_t)} \nonumber \\ 
 &\leq& v_i(a^*(\hat{\theta}_t), \theta_{a^*(\hat{\theta}_t)}) + \sum_{j \neq i}
v_j(a^*(\hat{\theta}_t), \theta_{a^*(\hat{\theta}_t)}) + \delta  \mathbb{E}_{\theta_{t+1} | a^*(\hat{\theta}_t),\hat{\theta}_t} 
W_{-i}(\theta_{-i,t+1}) \nonumber \\ 
 && \qquad - \ W_{-i}(\theta_{-i,t}) + \delta \mathbb{E}_{\theta_{t+1} | a^*(\hat{\theta}_t),\theta_t} \left[ W(\theta_{t+1}) - W_{-i}(\theta_{-i,t+1}) \right]
\end{eqnarray}

We also note that,
\begin{equation}
 \label{indep-tt}
 \mathbb{E}_{\theta_{t+1} | a^*(\hat{\theta}_t),\hat{\theta}_t} 
W_{-i}(\theta_{-i,t+1}) = \mathbb{E}_{\theta_{t+1} |
a^*(\hat{\theta}_t),\theta_{t}}  W_{-i}(\theta_{-i,t+1}) 
\end{equation}

This is because when $i$ is removed from the system in SA domain (while computing
$W_{-i}(\theta_{-i,t+1})$), the values of none of the other agents will depend on the
type $\theta_{i,t+1}$, due to PIV. And due to the independence of type transitions, $i$'s
reported type $\hat{\theta}_{i,t}$ can only influence $\theta_{i,t+1}$. Hence,
the reported value of agent $i$ at $t$, i.e., $\hat{\theta}_{i,t}$ cannot affect
$W_{-i}(\theta_{-i,t+1})$. 
%

Hence, Equation~\ref{eqn:exp-utility-final} can be rewritten to show the following inequality.
\begin{eqnarray}
 \label{replace}
 \lefteqn{u^{\text \mech}_i((\hat{\theta}_{i,t}, \theta_{-i,t}), (\hat{v}_{i,t}, v_{-i}(a^*(\hat{\theta}_{i,t}, \theta_{-i,t}), \theta_{a^*(\hat{\theta}_{i,t}, \theta_{-i,t})})) | \theta_t)} \nonumber \\ 
 &\leq& v_i(a^*(\hat{\theta}_t), \theta_{a^*(\hat{\theta}_t)}) + \sum_{j \neq i}
v_j(a^*(\hat{\theta}_t), \theta_{a^*(\hat{\theta}_t)}) + \cancel{{\color{blue} \delta  \mathbb{E}_{\theta_{t+1} |
a^*(\hat{\theta}_t),\theta_t}  W_{-i}(\theta_{-i,t+1})}} 
\nonumber \\  
 && - W_{-i}(\theta_{-i,t}) + \delta \mathbb{E}_{\theta_{t+1} | a^*(\hat{\theta}_t), \theta_t}
[W(\theta_{t+1}) - \cancel{{\color{blue} W_{-i}(\theta_{-i,t+1})}}] \ \mbox{(from Eq.~\ref{indep-tt})} \nonumber \\ 
 &=& \sum_{j \in N} v_j(a^*(\hat{\theta}_t), \theta_{a^*(\hat{\theta}_t)}) + \delta
\mathbb{E}_{\theta_{t+1} | a^*(\hat{\theta}_t), \theta_t} W(\theta_{t+1}) - W_{-i}(\theta_{-i,t}) \nonumber \\
 &\leq& \sum_{j \in N} v_j(a^*(\theta_{t}), \theta_{a^*(\theta_t)}) + \delta
\mathbb{E}_{\theta_{t+1} | a^*(\theta_{t}), \theta_t} W(\theta_{t+1}) - W_{-i}(\theta_{-i,t}) \nonumber \\
  &&\quad \mbox{(by definition of $a^*(\theta_{t})$, Eq.~\ref{allocation-MATRIX})} \nonumber \\ 
 &=& u^{\text \mech}_i(\theta_t, (v_i(a^*(\theta_t), \theta_{a^*(\theta_t)}), v_{-i}(a^*(\theta_t), \theta_{a^*(\theta_t)})) | \theta_t). \label{eq:utility-def} 
\end{eqnarray}
This shows that utility of agent $i$ is maximized when $\hat{\theta}_{i,t} = \theta_{i,t}$ and $\hat{v}_{i,t} = v_i(a^*(\theta_t), \theta_{a^*(\theta_t)})$. This proves that \mech\ is within period ex-post incentive compatible.

We now argue that the second stage is strictly EPIC for an agent $i$. This happens because of the quadratic penalty term $\left ( \hat{v}_{i,t} - v_i(a^*(\hat{\theta}_t), \hat{\theta}_{a^*(\hat{\theta}_t)}) \right )^2$ in the payment $p_i^*$ (Eq.~(\ref{payment-MATRIX})). Notice that if all the agents except $i$ report the types and values truthfully, and agent $i$ also reports her type truthfully in the first stage, then the penalty term will always penalize her if $\hat{v}_{i,t}$ is different from $v_i(a^*(\theta_t), \theta_{a^*(\hat{\theta}_t)})$, which is her true valuation. Hence, the best response of agent $i$ would be to report the true values in the second stage, which makes \mech\ strictly EPIC in this stage.
\end{proof}

\subsection{Why a dynamic pivot mechanism would not work in this setting}
\label{sec:why-not-DPM}

It is interesting to note that, if we tried to
use the dynamic pivot mechanism (DPM),
\citep{bergemann-valimaki10dynamic-pivot}, unmodified in this setting,
the true type profile $\theta_{t}$ in the first summation of
Eq.~(\ref{eqn:exp-utility-final}) would have been replaced by
$\hat{\theta}_{t}$, since this comes from the payment term
(Eq.~(\ref{payment-MATRIX})). The proof for the DPM relies on the \emph{private value} assumption (see the beginning of Section~\ref{model-back} for a definition) such that, when reasoning about the valuations for the other agents $j \neq i$, we have $v_j(a^*((\hat{\theta}_{i,t}, \theta_{-i,t})), (\hat{\theta}_{i,t}, \theta_{-i,t})) = v_j(a^*(\hat{\theta}_t), \theta_{j,t})$, with which the EPIC claim of DPM can be shown. But in the interdependent value setting, we cannot do such a substitution, and hence the proof of EPIC in DPM does not work. We have to invoke the second stage of value reporting in order to satisfy the EPIC.

\subsection{Ex-post individual rationality}

With SA and PIV, we now show that \mech\ is individually rational.

\begin{theorem}[Individual Rationality]
  \label{thm:epir}
  Under SMTT, with SA and PIV, \mech\ is w.p.\ EPIR.
\end{theorem}

\begin{proof}
Due to SA, the set of allocations excluding agent $i$, denoted by $A_{-i}=2^{N \setminus \{i\}}$, is already contained in the set of allocations including $i$, denoted by $A=2^N$. Formally, this means $a_{-i} \in A_{-i} \subseteq A \ni a$. Therefore, the policies $\pi_{-i} \in A_{-i}^\infty \subseteq A^\infty \ni \pi$.
%
Hence in the ex-post Nash
equilibrium, the utility of agent $i$ is given by,
\begin{align*}
 \lefteqn{u^{\text \mech}_i(\theta_t, (v_i(a(\theta_t), \theta_t), v_{-i}(a(\theta_t), \theta_t)) | \theta_t)} \\ 
 &= \sum_{j \in N} v_j(a^*(\theta_{t}), \theta_{t}) + \delta \mathbb{E}_{\theta_{t+1} | a^*(\theta_{t}), \theta_t} W(\theta_{t+1}) - W_{-i}(\theta_{-i,t}) \\ 
 &= W(\theta_{t}) - W_{-i}(\theta_{-i,t}) \\
 &\geq 0.
\end{align*}
The first equality comes from the last equality in Equation~\ref{eq:utility-def} and the second equality is by definition of $W(\theta_t)$ and $a^*(\theta_{t})$. The last inequality is an immediate consequence of SA and PIV, as the allocation that maximizes the social welfare excluding agent $i$ is already in the potential allocations when $i$ is present.
This proves that \mech\ is within period ex-post individually rational. 

\end{proof}


\subsection{Complexity of computing the allocation and payment} \label{sec:complexity}
The non-strategic version of the resource to task assignment problem was that of solving an MDP, whose complexity was polynomial in the size of state-space \citep{ye05complexity-of-MDP}. Interestingly, for the proposed mechanism, the allocation and payment decisions are also solutions of MDPs (Equations~\ref{allocation-MATRIX}, \ref{payment-MATRIX}). Hence the proposed mechanism \mech\ has polynomial time complexity in the number of agents and size of the state-space, which is the same as that of the dynamic pivot mechanism~\citep{bergemann-valimaki10dynamic-pivot}.


\section{Discussions on a General Result}
\label{sec:general}

We can generalize the assumptions of SA and PIV in the following way that would result in the same conclusions as in this paper. These definitions also serve to show the minimal requirements of the proofs.

Consider a set of all possible allocations denoted by $\mathcal{A}$. The valuations are called {\em independent of irrelevant agents (IIA)} with respect to a set of allocations $A \subseteq \mathcal{A}$ if for all $i \in N$, $\exists \ A_{-i} \subseteq \mathcal{A}$ s.t. $\forall \ a_{-i} \in A_{-i}$,
\begin{align*}
 v_j(a_{-i}, \theta) &= v_j(a_{-i}, \theta_{-i}) \\
 v_i(a_{-i}, \theta) &= 0
\end{align*}

SA and PIV together constitute a special case of IIA valuations. However, there exist not-so-restrictive examples as well. Consider a set of agents $N = \{1,2,\ldots, n\}$ having types $\theta_1, \theta_2, \ldots, \theta_n$ and a {\em dummy} agent $D$ who does not have any type. Let $\mathcal{A} = 2^{N \cup \{D\}}$ and $A = 2^N$. We define $A_{-i} = 2^{N \cup \{D\} \setminus \{i\}}$, the power set of agents where the dummy replaces agent $i$. Since the dummy does not have any type, the valuations of other agents after replacing agent $i$ with $D$ depends only on $\theta_{-i}$. Note, in particular, that $A_{-i} \nsubseteq A$.

If now, in addition, $A_{-i} \subseteq A$, then the allocations are called {\em monotone}. SA is a monotone set of allocations and PIV is IIA over that.

We can show that Theorem~\ref{thm:epic} extends with IIA valuations and Theorem~\ref{thm:epir} extends with IIA valuations with respect to monotone allocations. We omit the proofs since they follow identical arguments.

\section{Conclusions and Future Work}
\label{concl}

This paper provides a first attempt of designing a dynamic mechanism that is {\em strict} ex-post incentive compatible and efficient in an interdependent value setting with Markovian type evolution. In a restricted domain, which appears often in real-world scenarios, we show that our mechanism is ex-post individually rational as well. This mechanism, \mech, extends the mechanism proposed by \citet{mezzetti2004mechanism} to a dynamic setting and connects it to the mechanism proposed by \citet{bergemann-valimaki10dynamic-pivot}.

We have discussed the interesting and challenging domain of mechanism design with dynamically varying types and interdependent valuations. There has been very little work where dynamic types and interdependent values have been addressed together. Hence, there is very little known on the limits of achievable properties in this domain. We have provided one mechanism, namely \mech, that is w.p.\ EPIC, strict in the second stage, and under a restricted domain, even w.p.\ EPIR. However, we do not know what mechanism characterizes those properties in this domain. For example, a question that may arise is ``Is this the only efficient dynamic mechanism that satisfies strict w.p.\ EPIC in an interdependent value setting?''. For the static setting with independent values we have the Green-Laffont characterization result that answers this question for efficiency and DSIC. However, such a characterization result is absent for interdependent valuations for both static and dynamic 
mechanisms. 
Developing such a full characterization would be worthwhile.

\subsection*{Acknowledgements}
We are grateful to Ruggiero Cavallo, David C. Parkes, two anonymous referees and the associate editor for useful comments on the paper. This work was done when the first author was a student at the Indian Institute of Science and was supported by Tata Consultancy Services (TCS) Doctoral Fellowship. This work is part of a collaborative project between Xerox Research and Indian Institute of Science on incentive compatible learning.


\bibliographystyle{ecta}
\bibliography{master08042013}

\end{document}